\documentclass[11pt,a4paper]{article}
\usepackage{amsthm,amsmath,amssymb,mathrsfs}
\usepackage{paralist,enumerate}
\usepackage{fullpage}
\usepackage{hyperref,color}
\usepackage{graphicx}
\usepackage{framed}
\usepackage{caption}
\usepackage{subcaption}
\usepackage{multirow}
\usepackage{times}

\hypersetup{
	pdfproducer={},
	colorlinks=true,
	linkcolor=blue,
	citecolor=blue
	}

\newcommand{\set}[1]{\{ #1 \}}
\renewcommand{\phi}{\varphi}
\renewcommand{\epsilon}{\varepsilon}

\def\fakesmallcaps{\footnotesize}

\newcommand{\G}{\mathcal{G}}
\newcommand{\A}{\mathcal{A}}
\newcommand{\B}{\mathcal{B}}
\newcommand{\C}{\mathbb{C}}

\usepackage{sectsty}
\subsectionfont{\normalsize}
\sectionfont{\large}
\makeatletter
\def\@seccntformat#1{\csname the#1\endcsname.\quad}
\makeatother

\theoremstyle{definition}
\newtheorem{defn}{Definition}
\newtheorem{prob}{Problem}

\theoremstyle{default}

\newtheorem{claim}{Claim}
\newtheorem{thm}{Theorem}
\newtheorem{lem}{Lemma}

\title{\Large { \sc 2048 is (PSPACE) Hard, but Sometimes Easy}}
\author{ Rahul Mehta\thanks{Department of Computer Science, 35 Olden St., Princeton University, Princeton, NJ, 08540.} \\ Princeton University \\ {\sf rahulmehta@princeton.edu}}
\date{\today}

\begin{document}

\maketitle
\newcommand{\moveset}{\set{\Uparrow, \Downarrow, \Leftarrow, \Rightarrow}}
\begin{abstract}
We prove that a variant of 2048, a popular online puzzle game, is PSPACE-Complete. Our hardness result
holds for a version of the problem where the player has oracle access to the computer player's moves.
Specifically, we show that for an $n \times n$ game board $\mathcal{G}$, computing a 
sequence of moves to reach a particular configuration $\C$ from an initial
configuration $\C_0$ is PSPACE-Complete.
Our reduction is from Nondeterministic Constraint Logic (NCL).

We also show that determining whether or not there exists a fixed sequence of moves $\mathcal{S} \in \moveset^k$ of length $k$
that results in a winning configuration for an $n \times n$ game board is fixed-parameter tractable (FPT). We describe an algorithm
to solve this problem in $O(4^k n^2)$ time.
\end{abstract}

\section{Introduction}
\nocite{garey79}
\nocite{hearn05}
\nocite{demaine02}
\nocite{savitch70}
\nocite{rob97}

The online video game 2048 \cite{gcweb13} (read as ``twenty-forty-eight") has recently gained a great deal of popularity. Played on
a $4 \times 4$ game board with tiles containing positive powers of 2, the goal of the game is to combine equally-valued tiles to reach
the 2048 tile.

More specifically, the game board begins with an initial configuration of two tiles, of value $2$ or $4$, placed at arbitrary locations on the grid.
The player then selects to move {\sf UP}, {\sf DOWN}, {\sf LEFT}, or {\sf RIGHT} (denoted as $\moveset$). Each move shifts {\it all} tiles on the
grid in the direction chosen. If two adjacent tiles have the same value (i.e. $2^i$ for some $i > 0$), they combine, and the single resulting tile after the combination
will have value $2^{i+1}$. Following the player's move, the computer places a tile of value $2$ or $4$ at a random free location on the board.
Figure \ref{fig:game} outlines the first six moves of a game of 2048. Observe that if a tile's row or column in the direction of the move $x \in \moveset$ is unobstructed, 
the tile will slide as far as possible on the grid.

\begin{figure}[h]
\begin{center}
\begin{tabular*}{0.65\textwidth}{c c c}
\includegraphics[width=3.2cm]{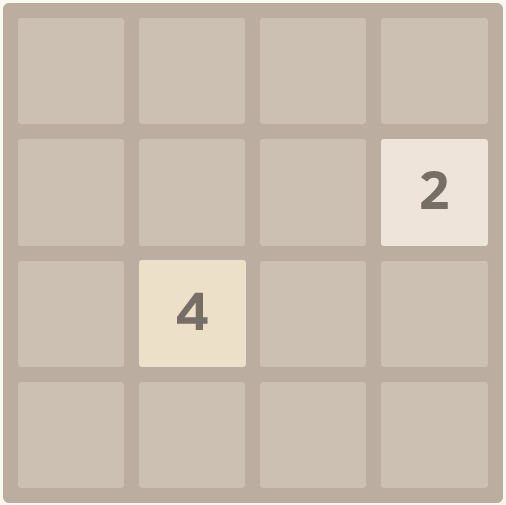} & \includegraphics[width=3.2cm]{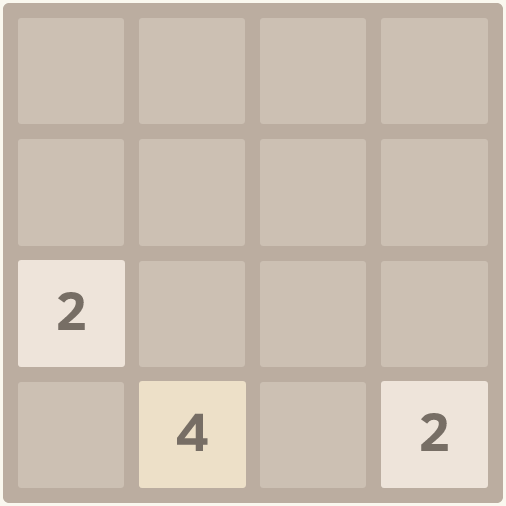} & \includegraphics[width=3.2cm]{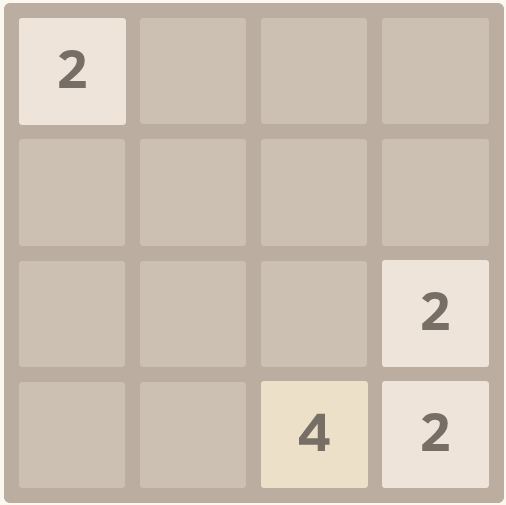} \\
{\sf Initial State} &  $\Downarrow$ & $\Rightarrow$  \\
\includegraphics[width=3.2cm]{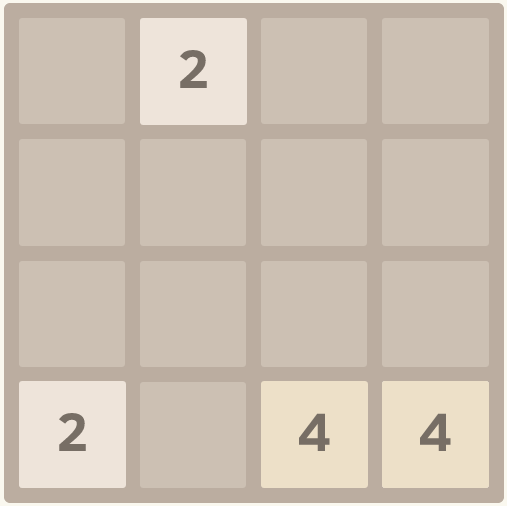} & \includegraphics[width=3.2cm]{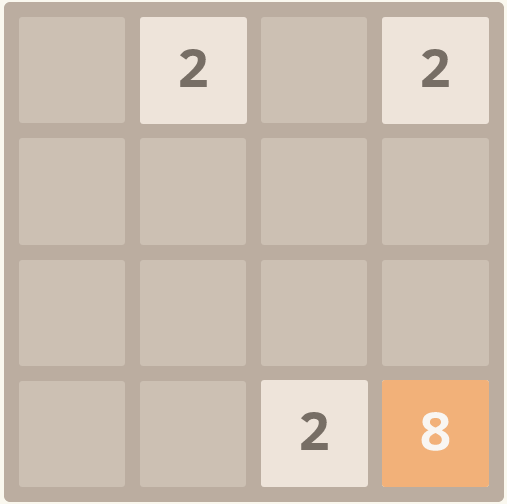} & \includegraphics[width=3.2cm]{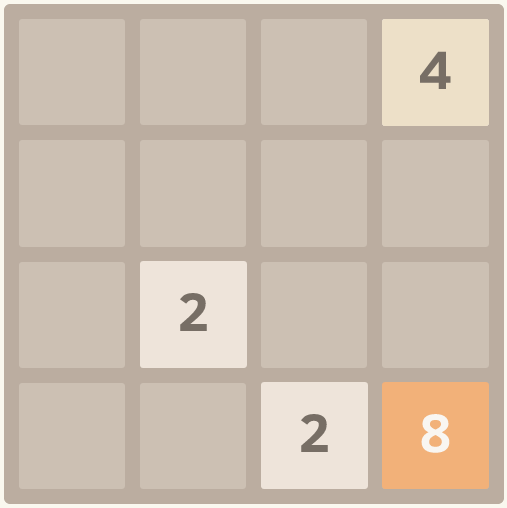} \\
$\Downarrow$ & $\Rightarrow$ & $\Rightarrow$
\end{tabular*}
\caption{The first six moves of a game of 2048.}\label{fig:game}
\end{center}
\end{figure}

Sliding tile games of a similar variety have been well-studied, both in the realm of puzzles as well as in algorithmic motion planning. Most recently,
Demaine and Hearn (\cite{hearn05, demaine02}) proved that several problems including solving sliding-block puzzles, Rush Hour, Sokoban, and Push-2-F
are all PSPACE-Complete. To accomplish this, they developed the Nondeterministic Constraint Logic (NCL) model of computation as a generic framework for PSPACE-Hardness
results (see \cite{hearn05}).

\subsection{Definitions}

We first define some basic concepts and terms that we will use throughout the paper, and then proceed to outline the two problems that will be studied in subsequent sections.
When we refer to a {\it game board} $\G$, we are referring to an $n \times n$ grid that can take on particular configurations, as described below;

\begin{defn}
Given a game board $\G$, a {\it configuration} of the board $\C:[n] \times [n] \rightarrow \set{2^\ell \mid \ell \in \mathbb{Z}_{> 0}} \cup \set{0}$ is a function that maps each of the $n^2$ locations on the game
board to a positive power of $2$, or $0$. For a particular configuration $\C$, a location $(i,j)$ is {\it empty} if $\C(i,j) = 0$.
\end{defn}

We describe locations in row-major order; for example, in the final (lower-right) configuration in Figure \ref{fig:game}, the tile of value 4 is located at $(1,4)$, the two tiles of value $2$ at $(3,2)$ and $(4,3)$, and the
tile of value $8$ at $(4,4)$.

We study two variants of 2048; first, we examine the complexity of computing a sequence of moves to reach a particular configuration $\C$ from an initial configuration $\C_0$, and second, 
we exhibit an efficient algorithm for determining whether a winning sequence of moves of length $k$ exists, starting from a particular initial configuration $\C_0$. In order to formalize these
two variants of the problem, we first recast 2048 as a two-player game between the human player $\A$ and the computer adversary $\B$;

\begin{defn}[The Game]
Given an initial configuration $\C_0$ of an $n \times n$ game board $\G$, consider the following game
between adversaries
$\A$ and $\B$, with a {\it goal configuration} $\C_f$. The game is played as follows;
\begin{compactenum}[\hspace{6mm}(1)]
	\item $\A$ makes a move $x \in \moveset.$ 
	\item $\B$ places one or more tiles of value $2^{\ell_1},\ldots,2^{\ell_k}$ ($\ell_1,\ldots,\ell_k \in \mathbb{Z}_{> 0}$) at locations $(i_1,j_1),\ldots,(i_k,j_k).$ 
\end{compactenum}
This game is played until one of two conditions occur; (1) the goal configuration $\C_f$ is reached, and $\A$ wins, or (2) there are no more moves that $\A$ can make (i.e. for any
$x \in \moveset$, applying $x$ will not change the configuration of the game board). In this case, $\B$ wins. Such a configuration is shown in Figure \ref{fig:over}
\end{defn}

\begin{figure}[h]
\begin{center}
\includegraphics[width=5cm]{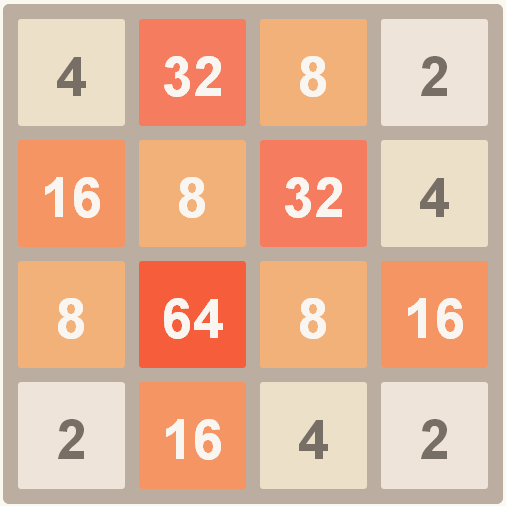}
\end{center}
\setcounter{figure}{1}
\caption{An example of a configuration satisfying (2) (game over).}\label{fig:over}
\end{figure}

\subsection{Problems}
We now introduce the two problems that will be studied in this paper.
Both variants that we study give player $\A$ {\it perfect knowledge}; that is, given a configuration $\C$ and a move $x \in \moveset$, $\A$ has oracle access to $\B$'s move if $x$ is applied to a board $\G$ with 
configuration $\C$.
This is denoted as $\B(\C,x)$. The response is a set of 3-tuples of the form $\set{(i_1,j_1,\ell_1),\ldots,(i_k,j_k,\ell_k)}$, where for index $t$, $(i_t,j_t)$ is the location of the added tile of value $2^{\ell_t}$. Accordingly,
an instance of {\sc 2048-Game} is uniquely determined by (1) the initial configuration $\C_0$, and (2) the responses of the computer $\B$.
In a sense, this is an easier version of the more general problem -- it is interesting to note that our hardness result still holds despite the additional information available to $\A$.

\begin{prob}[\sc 2048-Game]
Given an $n \times n$ game board $\mathcal{G}$ and an oracle $\B$ to the computer player, does there exist a sequence of moves $\mathcal{S} \in \moveset^*$ such that $\C$ is reached from $\C_0$?
\end{prob}

In Section \ref{sec:reduction}, we prove the following theorem:

\begin{claim}[Theorem \ref{thm:game_pspace}]
{\sc 2048-Game} is PSPACE-Complete.
\end{claim}

%
%


We then examine the related problem of determining whether or not there exists a fixed-length sequence of moves that results in a winning configuration.

\begin{prob}[\sc 2048-$k$-Moves]
Given an $n \times n$ game board $\G$, an oracle $\B$ to the computer player, and a goal tile $2^m$ (for some $m > 0$), 
does there exist a sequence of moves $\mathcal{S} \in \moveset^k$ of length $k$ such that after $k$ moves, the board is in a winning configuration $\C_f$ (that is, $\C_f(i,j) = 2^m$, for some $i$ and $j$)?
\end{prob}

When the problem is paramaterized in this manner (limiting the number of moves to some constant $k > 0$), we show that in fact {\sc 2048-$k$-Moves} is fixed-parameter tractable (FPT).
In Section \ref{sec:game_fpt}, we prove the following theorem:

\begin{claim}[Theorem \ref{thm:game_fpt}]
Given an $n \times n$ game board $\mathcal{G}$, {\sc 2048-$k$-Moves} is solvable in $O(4^k n^2)$ time.
\end{claim}

Before concluding the section, we show that the value of the maximum tile is monotonically nondecreasing. For a configuration $\C$, we denote the tile of maximum value as $\max(\C)$.

\begin{lem}\label{lem:nondec}
For a game board $\G$, a configuration $\C_i$, and the subsequent configuration $\C_{i+1}$,  $\max(\C_{i}) \leq \max(\C_{i+1})$.
\end{lem}
\begin{proof}
Immediate from the combination rule of the game. When adjacent tiles of value $2^\ell$ combine, they form a single tile of value $2^{\ell+1}$. Consider two cases; (1) the maximum
tile does not change value, or (2) the maximum valued tile does change value. If (1) occurs, then we have equality, since $\max(\C_i)$ is left unchanged. For (2), let $\max(\C_i) = 2^\ell$, for some $\ell > 0$.
The only case in which the tile in question changes is if it combines with a tile of the same value. By our previous observation, two adjacent tiles of value $2^\ell$ will combine to form a single tile of value $2^{\ell+1}$.
Thus, we have $\max(\C_{i+1}) = 2^{\ell+1} > 2^\ell = \max(\C_i)$, which completes the proof.
\end{proof}

The remainder of the paper is organized as follows; Section \ref{sec:ncl} describes the Nondeterministic Constraint Logic (NCL) model of computation due to Demaine and Hearn \cite{hearn05}, and specifically, the framework
it provides for PSPACE-Hardness reductions. Section \ref{sec:reduction} describes in detail our reduction from NCL to {\sc 2048-Game}. Finally, Section \ref{sec:game_fpt} describes an efficient algorithm to solve
{\sc 2048-$k$-Moves}, which establishes that it is fixed-parameter tractable (FPT).

\section{Nondeterministic Constraint Logic (NCL)}\label{sec:ncl}
We will now outline the Nondeterministic Constraint Logic (NCL) model of computation, which was developed by Demaine and Hearn \cite{hearn05} as a general framework for proving PSPACE-Completeness results. 
All of their PSPACE-Completeness results are reductions from Quantified Boolean Formula (QBF), a well-known PSPACE-Complete problem \cite{garey79}.

An NCL {\it machine}
consists of a {\it constraint graph}, which is an arbitrary undirected graph $G = (V,E)$ with edge weights $w:E \rightarrow \set{1,2}$ (all edges have weight 1 or 2). A {\it configuration} of the machine is an orientation
of the edges of $G$. A configuration is {\it valid} if for each vertex $v \in V$, the sum of incoming edge weights is at least $2$. A move is made by reversing a single edge in the network such that the configuration remains
valid.

A natural question to ask is whether or not a particular edge can be reversed after a sequence of moves ({\sc Config-to-Edge}), or alternately, whether there exists a sequence of moves to reach a particular configuration from some initial configuration
({\sc Config-to-Config}).
Both of these problems are PSPACE-Complete (see \cite{hearn05}). This hardness result still holds when we restrict the vertices to those with incident edge weights 1,1, and 2, and those with 1, 1, and 1. These vertices
are illustrated in Figure \ref{fig:ncl_orig} (taken from \cite{hearn05}):

\begin{figure}[h]
\begin{center}
\begin{minipage}[b]{.25\textwidth}
\centering
\includegraphics[width=2.5cm]{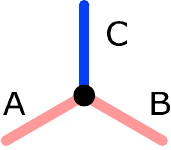}
\subcaption{{\sc And}}
\end{minipage}%
\begin{minipage}[b]{.25\textwidth}
\centering
\includegraphics[width=2.5cm]{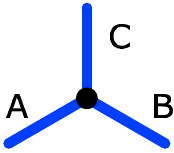}
\subcaption{{\sc Or}}
\end{minipage}
\end{center}
\caption{NCL {\sc And} and {\sc Or} Vertices}\label{fig:ncl_orig}
\end{figure}

Red edges have weight 1, and blue edges have weight 2. For an orientation of the edges, observe that for vertex (a), the blue edge can point out if and only if the two red edges are pointing in (to satisfy the previously-mentioned constraint
that the in-flow on all vertices must be at least 2). Similarly, for (b), the top blue edge can point outward if and only if one of the bottom edges is pointing in. Thus, we can clearly see that (a) functions as an {\sc And} gate of sorts
and (b) as an {\sc Or} gate.

Demaine and Hearn in fact strengthen this result even further, and show that {\sc Config-to-Edge} and {\sc Config-to-Config} both remain PSPACE-Complete when the constraint graph $G$ is planar. Thus, proving a game to be PSPACE-Hard
reduces to simply constructing NCL {\sc And} and {\sc Or} gadgets from the problem in question, and then demonstrating how to use them to construct arbitrary planar constraint graphs.

\section{The Reduction}\label{sec:reduction}

In this section, we outline the reduction from planar NCL to {\sc 2048-Game}. We begin by constructing NCL {\sc And} and {\sc Or} vertex gadgets with $4 \times 4$ subinstances of {\sc 2048-Game}, and then demonstrate how to connect them
together to make arbitrary planar constraint graphs. This section will prove the following theorem;

\begin{thm}\label{thm:game_pspace}
{\sc 2048-Game} is PSPACE-Complete.
\end{thm}

Before outlining the reduction, however, we must show that {\sc 2048-Game} is in fact contained in PSPACE. A rigorous proof of this fact is given below;

\begin{lem}\label{lem:game_in_pspace}
$\text{\sc 2048-Game} \in \text{PSPACE}$.
\end{lem}
\begin{proof}
We begin by giving an NPSPACE algorithm to decide {\sc 2048-Game}. For an $n \times n$ game board $\G$, we clearly can
represent the current configuration of the game board in polynomial space.
 Starting with an initial configuration $\C_0$ and a goal configuration $\C_f$, we can nondeterministically select a move $x \in \moveset$ at each step (without
consulting the oracle to $\B$), and update the previous configuration with the current one. We repeat the above procedure until one of the following conditions occurs;
\begin{enumerate}[\hspace{5.5mm}(1)]
	\item The goal configuration $\C_f$ is reached, in which case output {\sf YES}.
	\item There are no possible moves from the current configuration $\C$, and it is not the goal configuration (game over); output {\sf NO}.
	\item For the current configuration $\C$, $\max(\C) > \max(\C_f)$, in which case output {\sf NO}.
\end{enumerate}

Checking for these three conditions guarantees that the algorithm will terminate. Specifically, for condition (3), if there is a tile of value greater than the maximum-valued tile in the goal
configuration $\C_f$, then $\C_f$ is unreachable from the current configuration $\C$. This is immediate from Lemma \ref{lem:nondec}.
Thus, we have an NPSPACE algorithm for deciding {\sc 2048-Game}. This is easily converted to a PSPACE algorithm due to Savitch's Theorem \cite{savitch70}, completing our proof.
\end{proof}

\subsection{NCL A{\fakesmallcaps ND} and O{\fakesmallcaps R} Gadgets}

In order to prove Theorem \ref{thm:game_pspace}, we first construct NCL {\sc And} and {\sc Or} gadgets from small instances of {\sc 2048-Game}, and then show how to connect them into arbitrary planar
graphs. 

\begin{figure}
\begin{center}
\begin{minipage}[b]{.3\textwidth}
\centering
\includegraphics[width=4cm]{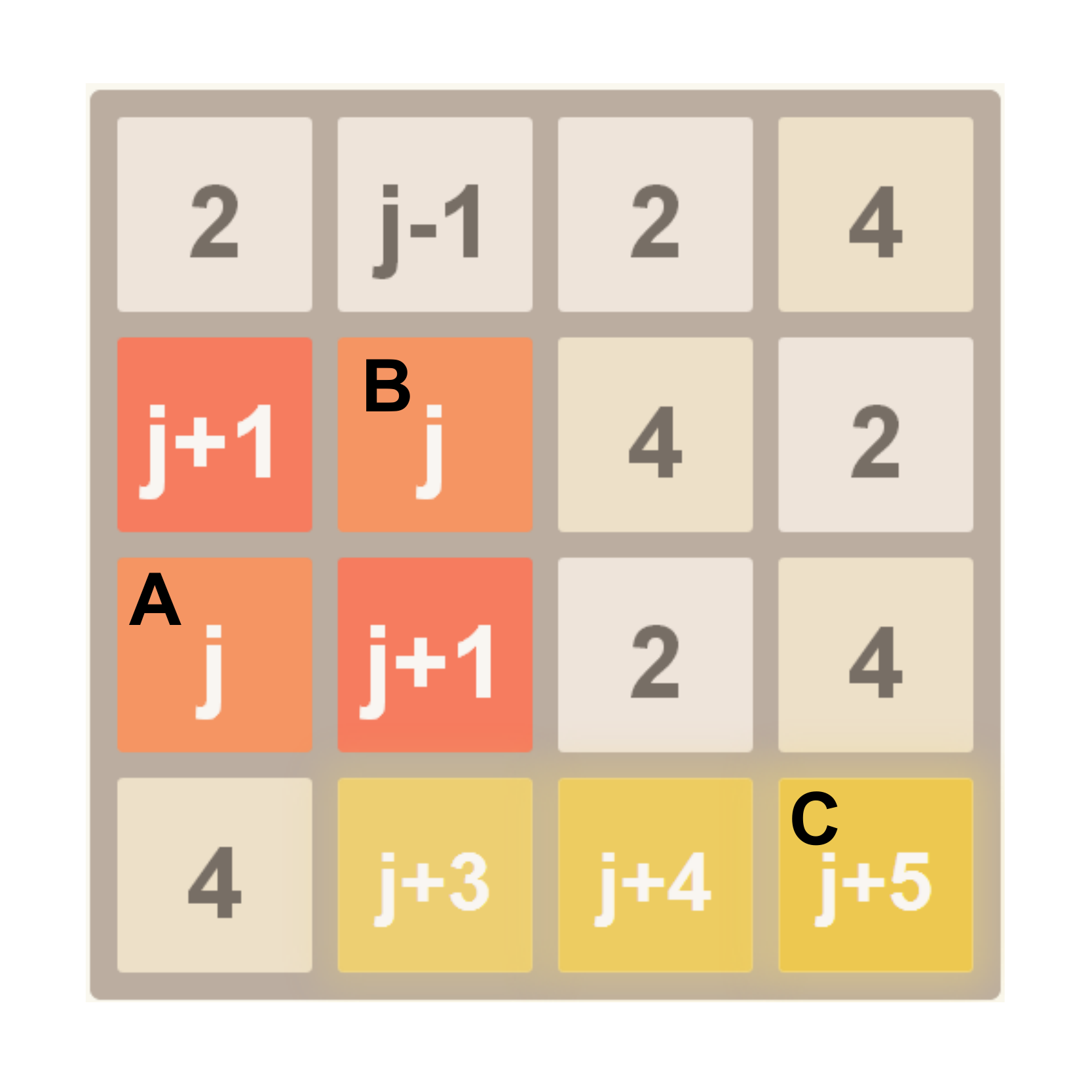}
\subcaption{$j$-{\sc And} Gadget}
\end{minipage}%
\begin{minipage}[b]{.3\textwidth}
\centering
\includegraphics[width=4.1cm]{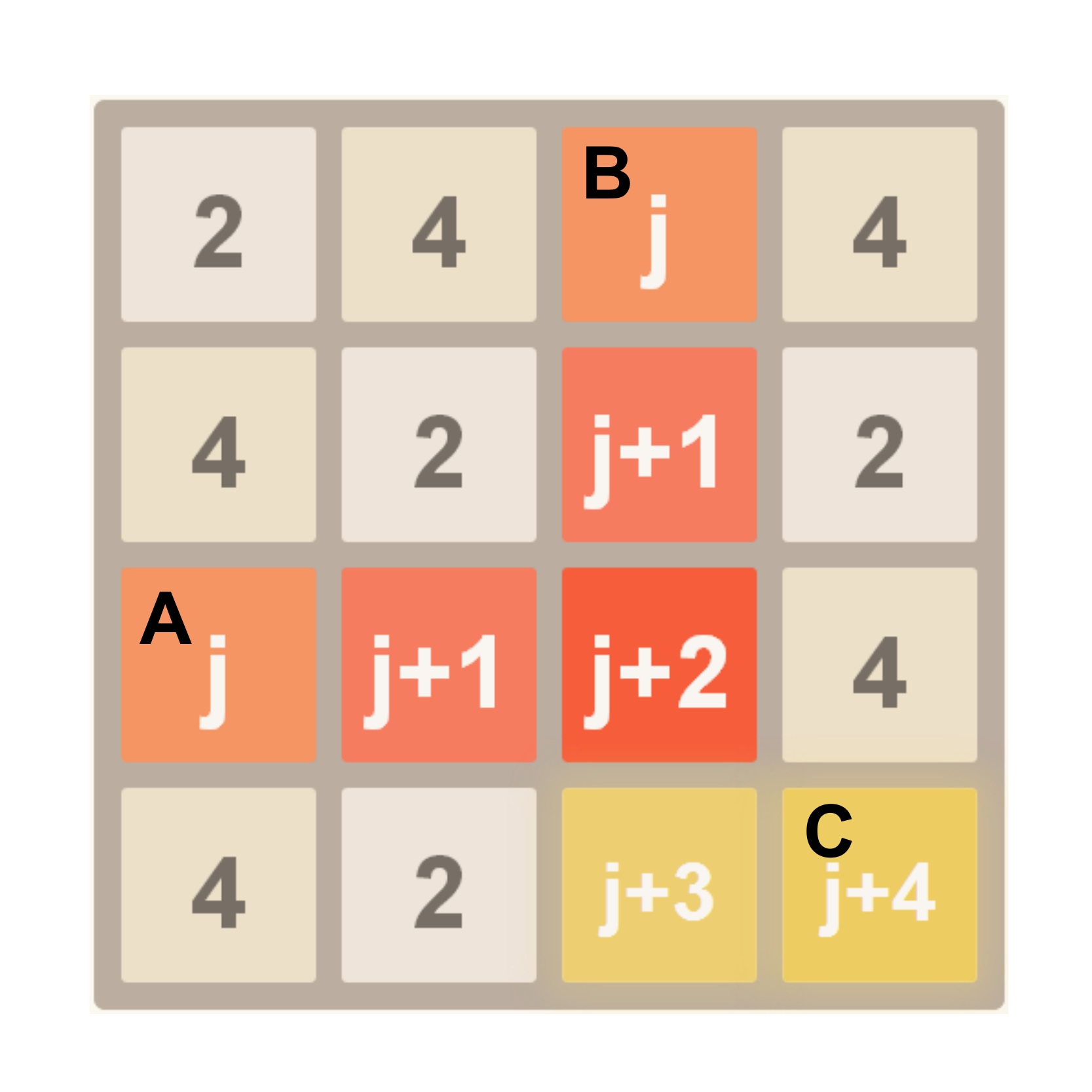}
\subcaption{$j$-{\sc Or} Gadget}
\end{minipage}
\begin{minipage}[b]{.3\textwidth}
\centering
\includegraphics[width=3.7cm]{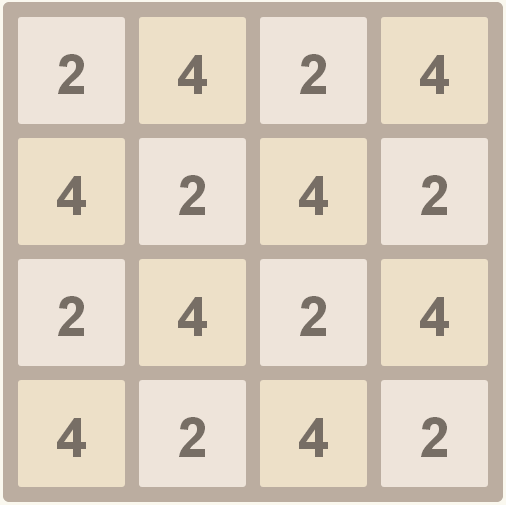}
\subcaption{$(2,4)$ Lattice}
\end{minipage}
\end{center}
\caption{NCL {\sc And} and {\sc Or} Vertices. Tiles with labels $j$, $j+1$, etc. correspond to values $2^j, 2^{j+1},$ etc. In both gadgets, connections {\sf A} and {\sf B} are facing {\it out}, while connection {\sf C} is facing {\it in}.}\label{fig:ncl_2048}
\end{figure}

Figure \ref{fig:ncl_2048} shows two examples of {\sc And} and {\sc Or} vertex gadgets. Since the tiles in 2048 are assigned a numerical value, an added difficulty presents itself when connecting vertex gadgets together (since the gadgets
rely on adjacent tiles containing specific powers of 2, including the connectors to ``activate" each vertex gadget). Thus, we introduce the notion of a $j$-{\sc Or} and $j$-{\sc And} gadget, as is shown in Figure \ref{fig:ncl_2048}. The labels
{\sf A, B,} and {\sf C} refer to the connection points for each gadget; these labels correspond to the edges for the NCL vertices in Figure \ref{fig:ncl_orig}.

A connection {\sf A} or {\sf B} is considered to be {\it activated} for a $j$-{\sc And} or $j$-{\sc Or} gadget if the tile has value not equal to $2^j$ --  this directly corresponds to reversing an edge in an NCL constraint graph. For both {\sc And} and {\sc Or} gadgets, {\sf A} and {\sf B} are facing {\it in} if the corresponding tile is activated, and {\it out} otherwise.

The tile labeled {\sf C} is also considered to be activated if its value has increased; thus, for a $j$-{\sc And} gadget, {\sf C} is activated if the tile a value not equal to $2^{j+5}$ or greater, and for a $j$-{\sc Or} gadget, if the tile has not equal to $2^{j+4}$. 
Accordingly, {\sf C} is facing {\it out} if it is activated, and {\it in} otherwise.

Before proving the correctness of the two gadgets, we outline the use for the gadget in Figure \ref{fig:ncl_2048}(c); the $(2,4)$ 
Lattice. This structure has an important property; for any move $x \in \moveset$, the configuration remains unchanged. Thus,
it is perfectly rigid. We will embed all of the other gadgets and connection pieces into the $(2,4)$ lattice so as to prevent
large portions of the game board shifting when a move is made by $\A$.

Additionally, whenever a square of the game board is uncovered by a move, the oracle $\B$, in this construction, will respond by placing a tile $t \in \set{2,4}$ in the vacant location,
depending on which tile correctly continues the $(2,4)$ lattice pattern. This is where the {\it perfect knowledge} comes into play; we are allowed to ``program" player $\B$'s responses
to various moves, which greatly simplifies parts of the reduction.



\begin{lem}\label{lem:jand_correct}
For a $j$-{\sc And} vertex gadget, {\sf C} can face {\it out} if and only if {\sf A} and {\sf B} are facing {\it in}.
\end{lem}
\begin{proof}
First, we observe that the tiles in the $j$-{\sc And} gadget will not shift at all unless either {\sf A} or {\sf B} is activated, for any move $x \in \moveset$, since there are no adjacent tiles of the same value.

Next, we demonstrate that for any sequence of moves which results in {\sf C} activating, {\sf A} and {\sf B} must both be activated. We claim that for {\sf C} to be activated in configuration $\C_i$,
 $\C_{i-3}(3,2) = 2^{j+3}$. That is, location $(3,2)$ must contain $2^{j+3}$ before {\sf C} can be activated. This is summarized in the diagram below; \\

\begin{center}
 \includegraphics{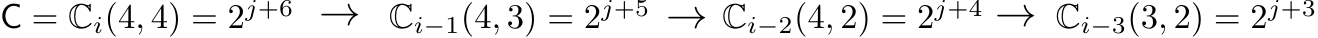}
\end{center}
 
 Correctness of the diagram can be shown inductively, by assuming that the value in $\C_{k-1}$ is necessary for the specified value in $\C_k$ to appear. Thus, proving the lemma reduces to determining
 the conditions under which there is a $\C$ such that $\C(3,2) = 2^{j+3}$. 
 
 Since we have already shown that {\sf A} or {\sf B} must be activated for any movement within the gadget to occur, it suffices to prove that only activating {\sf A} or {\sf B} does not suffice. Without loss of
 generality, assume that {\sf A} is activated but not {\sf B} (the proof of the opposite assumption follows the exact same reasoning). If only {\sf A} is activated, then for some $\C$, $\C(3,1) = 2^{j+1}$, and
 there exists a sequence of moves such that for some subsequent $\C'$, $\C'(3,1) = 2^{j+2}$. However, $\C'(3,2) = 2^{j+1}$, so no combination is possible to achieve $\C''(3,2) = 2^{j+3}$, as is desired.
 Thus, {\sf A} and {\sf B} must be activated for {\sf C} to be activated, which concludes the proof.

%
\end{proof}

An example of a sequence of moves to activate {\sf C} for a $4$-{\sc And} vertex gadget is in Appendix \ref{app:seq} (Figure \ref{fig:and_path}).
Now, we turn our attention to the $j$-{\sc Or} vertex. The proof closely follows that of Lemma \ref{lem:jand_correct}.

\begin{lem}\label{lem:jor_correct}
For a $j$-{\sc Or} vertex gadget, {\sf C} can face out if and only if {\sf A} or {\sf B} are facing {\it in}.
\end{lem}
\begin{proof}
Again, we see that tiles in the $j$-{\sc Or} gadget will not shift unless {\sf A} or {\sf B} are activated, for $x \in \moveset$, since no adjacent tiles in the gadget's initial configuration are of the same value.

Next, we show that for any sequence of moves resulting in the activation of {\sf C}, either {\sf A} or {\sf B} must be activated. 
Suppose {\sf C} is activated in configuration $\C_i$. Backtracking along the tiles in the gadget, we show that for {\sf C} to be activated in $\C_i$, then in $\C_{i-4}$, {\sf A} or {\sf B} must be activated. The diagram
below describes the required configurations. The correctness of the diagram is proven inductively by assuming that the value(s) described in $\C_{k-1}$ is a necessary condition for the value(s) in $\C_{k}$ to be reached.

{\centering
\includegraphics{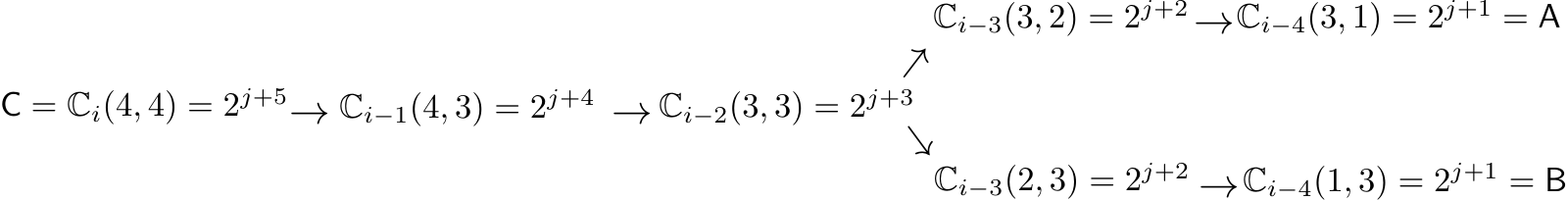}
} \\

The validity of the moves in the diagram can be verified by inspection of the $j$-{\sc Or} gadget in Figure \ref{fig:ncl_2048}(b). The fact that $\C_{i-4}(\text{\sf A}) = 2^{j+1}$ or $\C_{i-4}(\text{\sf B}) = 2^{j+1}$ must hold
for $\C_i(\text{\sf C}) = 2^{j+5}$ is immediate from the diagram; we conclude that {\sf C} can be activated, and thus face out, if and only if {\sf A} or {\sf B} is activated, and therefore facing in.
\end{proof}

An example of a sequence of moves to activate {\sf C} for a $4$-{\sc Or} vertex gadget is in Appendix \ref{app:seq} (Figure \ref{fig:or_path}).

\subsection{The Reversible NCL A{\fakesmallcaps ND}/O{\fakesmallcaps R} Gadget}
Now that we have described the NCL {\sc And} and {\sc Or} gadgets, we have one final vertex gadget left to define. Observe that for either the NCL $j$-{\sc And} or $j$-{\sc Or} gadget, once it is activated, the configuration cannot
be reversed; that is, once a gadget is set in a particular configuration, it cannot be altered. Thus, if the initial configuration of an NCL machine $\C_0$ contains an activated {\sc And} vertex, and the goal is to reverse one of
the incoming edges to the vertex, then this is impossible in our vertex gadget construction but possible in NCL.

To alleviate this, we construct a final gadget, namely the Reversible $j$-{\sc And}/{\sc Or} gadget. This gadget is placed in an instance of {\sc 2048-Game} when an {\sc And} vertex in the original NCL constraint graph is
activated, or when both edges of weight 1 are entering an {\sc Or} vertex. The conditions under which a Reversible $j$-{\sc And}/{\sc Or} gadget is used are outlined in Figure \ref{fig:and_or_cond}.

\begin{figure}[h]
\begin{center}
\includegraphics{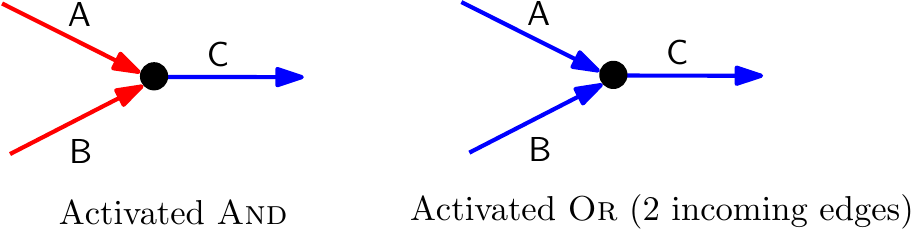}
\end{center}
\caption{The conditions under which a Reversible $j$-{\sc And}/{\sc Or} is used in place of a normal vertex gadget.}\label{fig:and_or_cond}
\end{figure}

We now introduce the gadget. {\sf A} and {\sf B} are considered to be activated if their value is not equal to $2^{j+4}$. Thus, in the initial configuration they are facing {\it in}, and when activated, are facing {\it out}. {\sf C} is activated when
its value is not equal to $2^{j}$; when it is activated, it is facing {\it in}, and when it is not, it is facing {\it out}. The gadget is defined in Figure \ref{fig:reverse_and_or}.

\begin{figure}[h]
\begin{center}
\includegraphics[width=4.75cm]{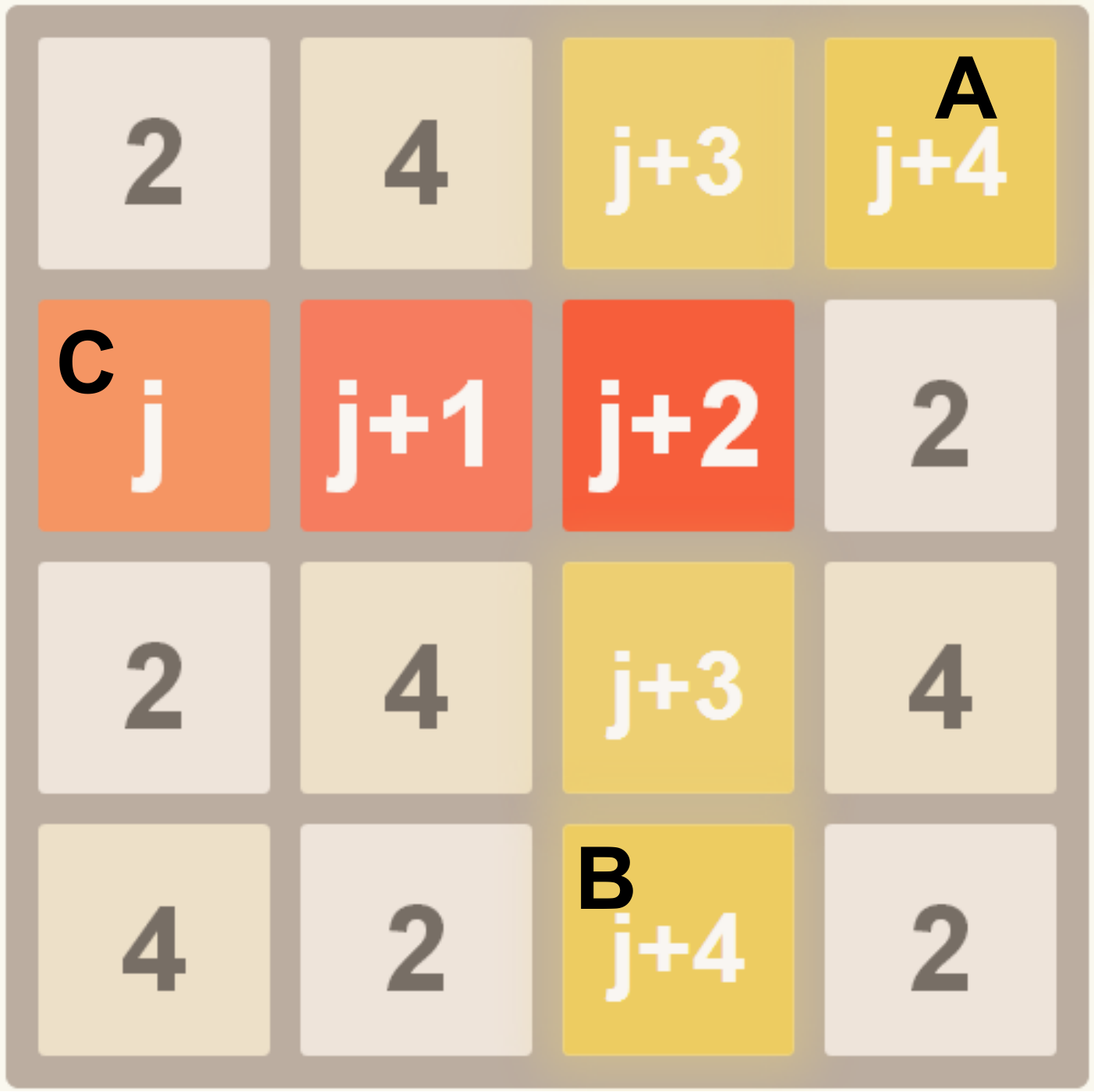}
\end{center}
\caption{The Reversible $j$-{\sc And}/{\sc Or} vertex gadget.}\label{fig:reverse_and_or}
\end{figure}

One specific that we note about the gadget is it's behavior after the tile at location $(2,3)$ (in the initial configuration, $2^{j+2}$) is shifted from that location, the computer player $\B$ places a new tile of value $2^{j+3}$ at $(2,3)$.
That way, both {\sf A} and {\sf B} can be activated, and therefore face out. However, it is also possible for either {\sf A} or {\sf B} to face out, but not the other. In fact, we prove that all possible configurations of the vertex gadgets
in Figure \ref{fig:and_or_cond} are reachable with the Reversible $j$-{\sc And}/{\sc Or} gadget. These reachable configurations are listed in Figure \ref{fig:reachable} below;

\begin{figure}[h]
\begin{center}
\includegraphics{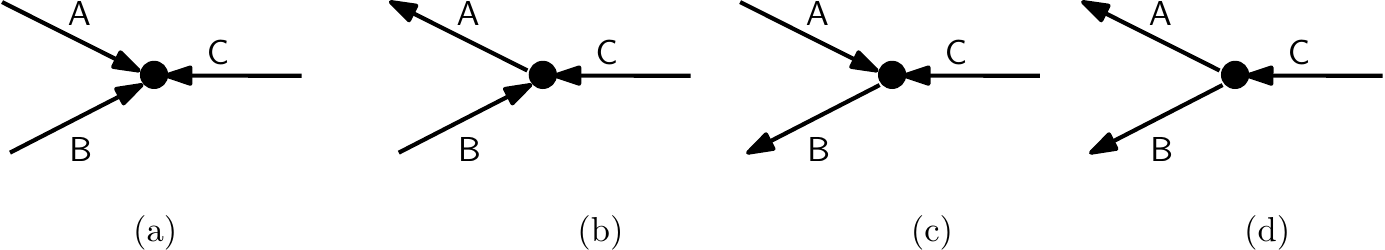}
\end{center}
\caption{Possible configurations reachable by Activated {\sc And} and {\sc Or} vertices with two incoming edges.}\label{fig:reachable}
\end{figure}

We now prove that these configurations are reachable by exhibiting sequences of moves;

\begin{lem}\label{lem:reach_rev_and_or}
The configurations listed in Figure \ref{fig:reachable} are reachable for any given Reversible $j$-{\sc And}/{\sc Or} gadget.
\end{lem}
\begin{proof}
The following diagram depicts the Reversible $j$-{\sc And}/{\sc Or} gadget in the four configurations listed in Figure \ref{fig:reachable}. The proof of the lemma below contains the sequences of moves necessary to reach
each configuration; this can be verified by inspection. The diagram shows the configurations for a Reversible $j$-{\sc And}/{\sc Or} gadget.

\begin{center}
\begin{tabular}{c c c c}
\includegraphics[width=3.2cm]{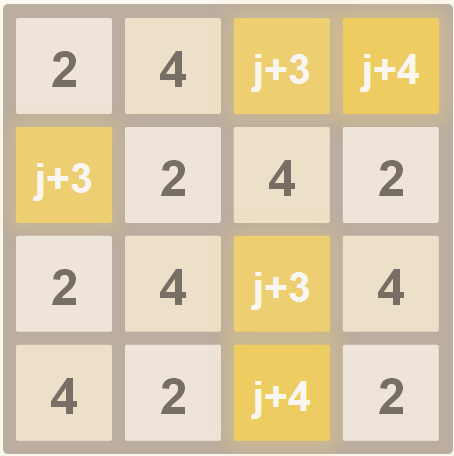} & \includegraphics[width=3.2cm]{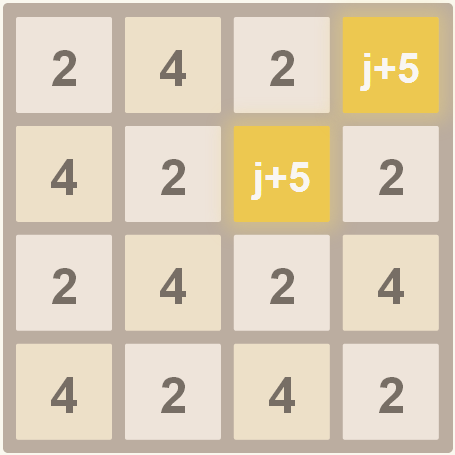} & \includegraphics[width=3.2cm]{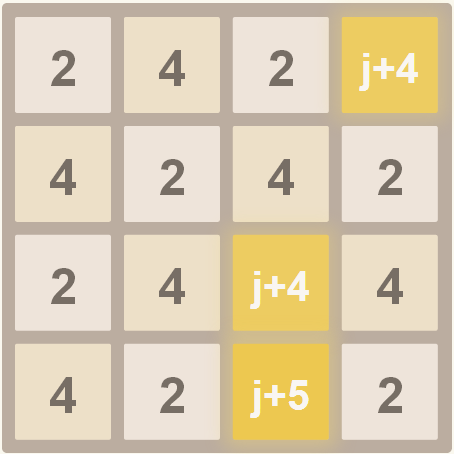} & \includegraphics[width=3.2cm]{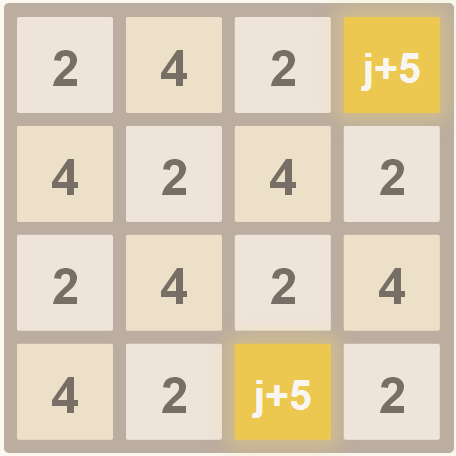} \\
(a) & (b) & (c) & (d)
\end{tabular}
\end{center}

We now give sequences of moves to achieve the given configurations. The sequences can easily be verified by inspection of the gadget;
\begin{inparaenum}[(a)]
\item $\Rightarrow, \Leftarrow, \Leftarrow$
\item $\Rightarrow, \Rightarrow, \Rightarrow, \Uparrow, \Rightarrow, \Uparrow, \Uparrow$
\item $\Rightarrow, \Rightarrow, \Rightarrow, \Downarrow, \Downarrow$
\item $\Rightarrow, \Rightarrow, \Rightarrow, \Uparrow, \Rightarrow, \Downarrow, \Downarrow$.
\end{inparaenum}
Thus, the configurations in Figure \ref{fig:reachable} are all reachable by the Reversible $j$-{\sc And}/{\sc Or} vertex, completing
the proof.
\end{proof}

\subsection{Constructing Arbitrary Planar Constraint Graphs}\label{sec:planar}
We now describe how to construct arbitrary planar graphs from the NCL vertex gadgets. The vertex gadgets, as well as the connection pieces described
in this section, are embedded in an arbitrarily large game board $\G$ comprised of $4 \times 4$ sub-instances. Such a layout is shown in Figure \ref{fig:board}

\begin{figure}[h]
\begin{center}
\includegraphics[width=5cm]{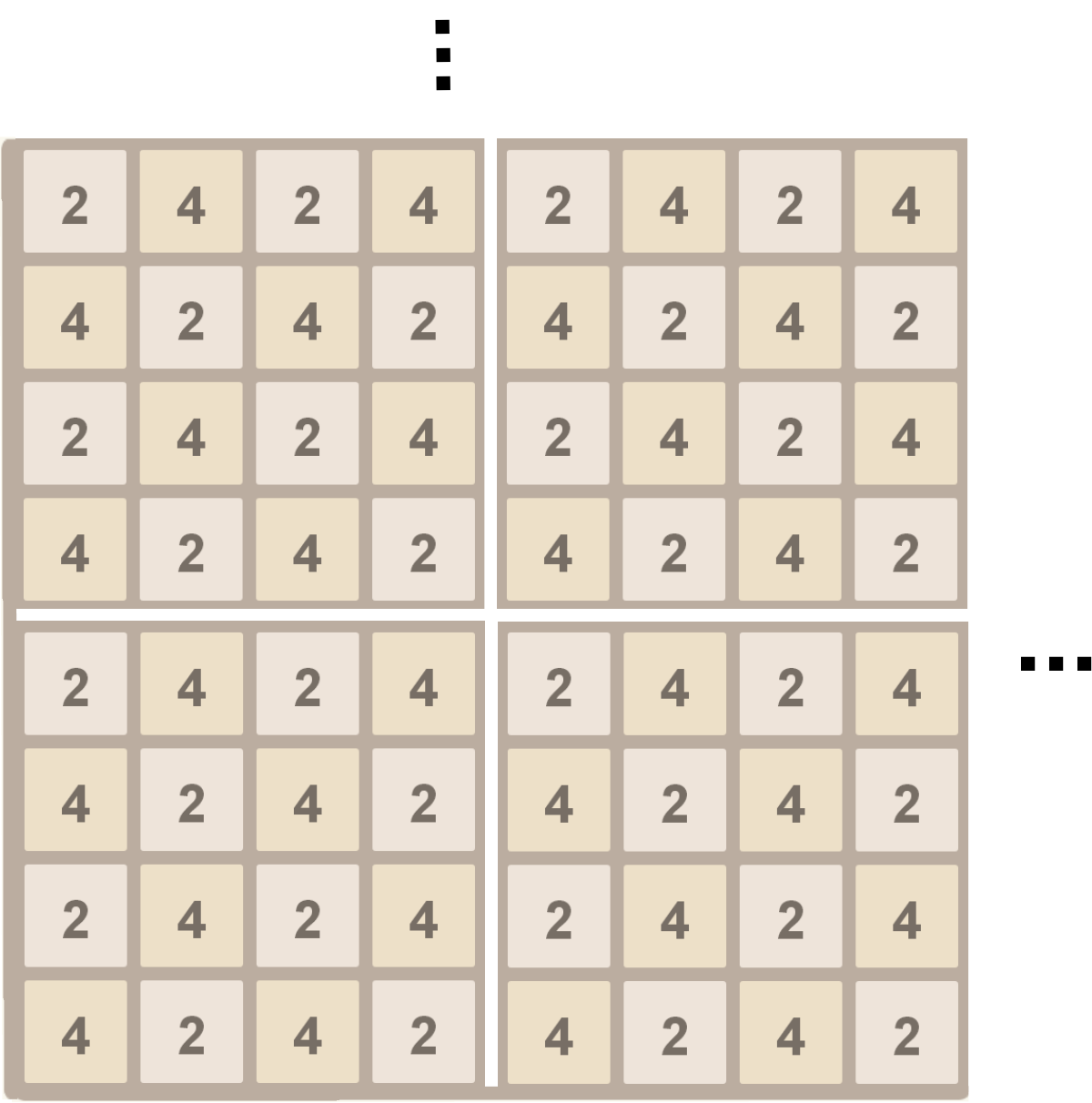}
\end{center}
\caption{An arbitrarily large game board $\G$ comprised of $4 \times 4$ sub-instances of {\sc 2048-Game}.}\label{fig:board}
\end{figure}


There is also a second issue that we must navigate. Due to the vertex gadgets' connection
tiles ({\sf A}, {\sf B}, and {\sf C}) containing numerical values of $2$, connecting arbitrary vertices is not as simple as in
other reductions from NCL. 

We address this by turning to the famed Four Color Theorem (see \cite{mat09} for more detail regarding several proofs).
Since we are constructing {\it planar} NCL constraint graphs, we can turn to the Four Color Theorem, and in particular, an $O(n^2)$ algorithm
by Robertson, et. al. \cite{rob97} to find a four coloring. By assigning colors to the vertices in the constraint graph, we are able to specify the
value of $j$ for the $j$-{\sc And}, $j$-{\sc Or}, and Reversible $j$-{\sc And/Or} gadgets. For each of the four colors, let $j = 4, 5, 6, 7$ respectively.
Note that the diagrams in Appendix \ref{app:seq} are for NCL gadgets with $j=4$.

\paragraph{The $k$-$k'$-C{\fakesmallcaps ONNECTION} Gadget}
We now describe a connection piece to connect gadgets with differing values of $j$. We call this a $k$-$k'$-{\sc Connection} piece, and its layout in our
context of the $4 \times 4$ sub-instance of {\sc 2048-Game}. The diagram is contained in Figure \ref{fig:con}

\begin{figure}[h]
\begin{center}
\includegraphics[height=4.5cm]{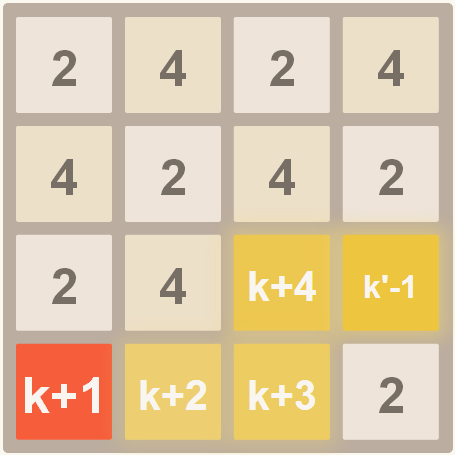}
\end{center}
\caption{The $k$-$k'$-{\sc Connection} gadget.}\label{fig:con}
\end{figure}

The gadget is attachable to a gadget that will terminate with value $2^k$ on any one of its output edges ({\sf C} for normal {\sc And/Or} gadgets, and {\sf A} and {\sf B} for the Reversible {\sc And/Or} gadget).
We note that this gadget can be shifted vertically within the $4 \times 4$ sub-instance to yield connection pieces for output tiles at varying locations. 

The major ``trick," so to speak, lies in the gadget's activation sequence. When the tile with initial value of $k+4$ at $(3,3)$ vacates that location, the computer player $\B$ will place a tile of value $k'-1$
at $(3,3)$. Then, the tile is free to combine with the second tile of value $k'-1$ at $(3,4)$, which will result in $(3,4)$ having the value $k'$, which in turn implies that it will activate a gadget with $j=k'$.

\begin{lem}\label{lem:con}
The $k$-$k'$-{\sc Connection} gadget contains an activation sequence such that an outgoing tile of value $k$ will activate a gadget with $j=k'$.
\end{lem}
\begin{proof}
We will exhibit an activation sequence of the gadget to demonstrate that it indeed fulfills its promise, namely that it will activate an NCL vertex gadget with $j=k'$ if the outgoing tile of the previous
gadget has value $k$. This is exhibited with values $k=7$ and $k'=14$.  \\

\begin{tabular}{c c c c c}
\includegraphics[width=2.5cm]{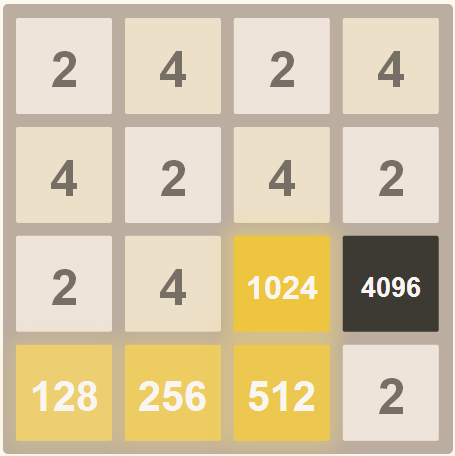} & \includegraphics[width=2.5cm]{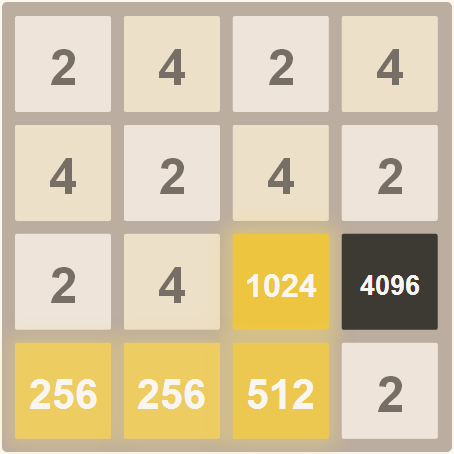} & \includegraphics[width=2.5cm]{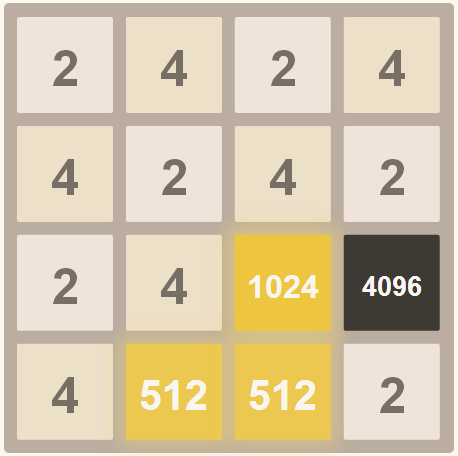} & \includegraphics[width=2.5cm]{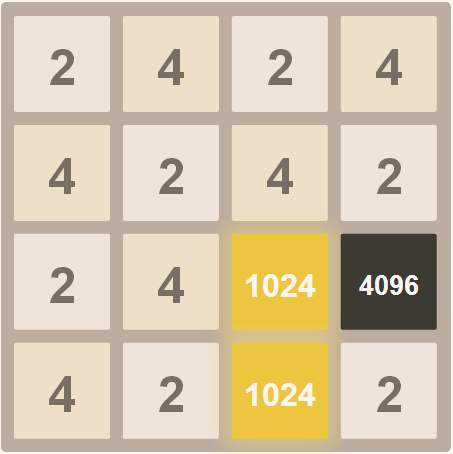}  & \includegraphics[width=2.5cm]{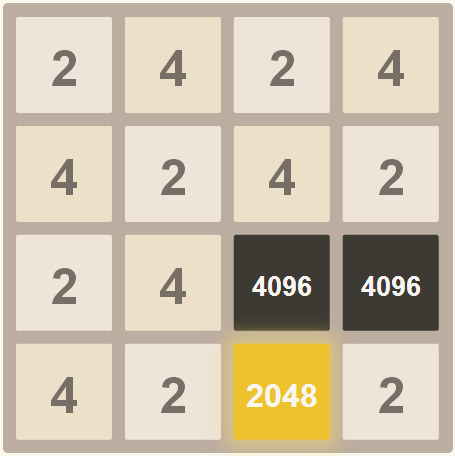}\\ 
{\sf Initial State} & $\Rightarrow$ & $\Rightarrow$ & $\Rightarrow$ & $\Downarrow$ \\
 \multicolumn{5}{c}{\includegraphics[width=2.5cm]{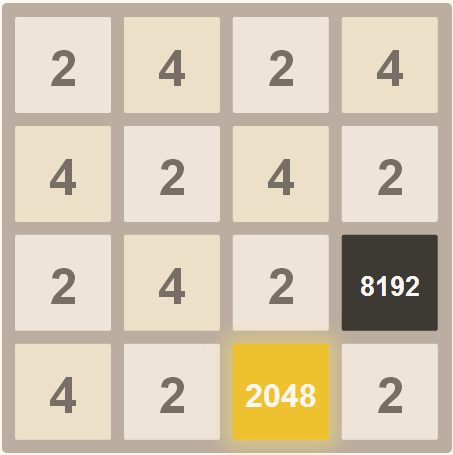}} \\  \multicolumn{5}{c}{$\Rightarrow$} \\
\end{tabular}

\end{proof}

\paragraph{The $k$-L{\fakesmallcaps INE} and $k$-C{\fakesmallcaps ORNER} Pieces} We finally describe two additional connection pieces which will finish our construction
of arbitrary planar graphs. These are the {\sc $k$-Line} and {\sc $k$-Corner} gadgets, and are outlined below in Figure \ref{fig:lin_corn};

\begin{figure}[h]
\begin{center}
\begin{minipage}[b]{.25\textwidth}
\centering
\includegraphics[width=3cm]{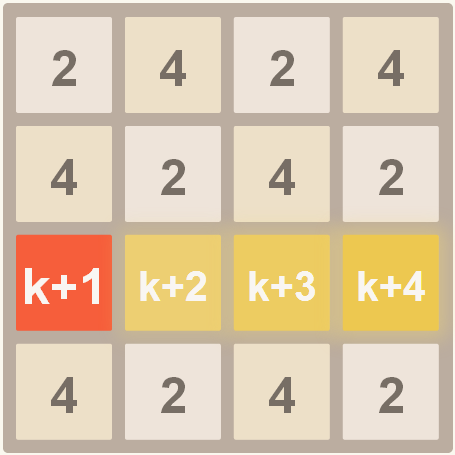}
\subcaption{{\sc $k$-{\sc Line}}}
\end{minipage}%
\begin{minipage}[b]{.25\textwidth}
\centering
\includegraphics[width=3cm]{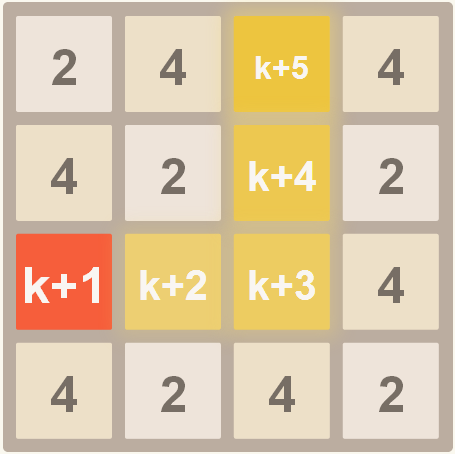}
\subcaption{{\sc $k$-Corner}}
\end{minipage}
\end{center}
\caption{$k$-{\sc Line} and $k$-{\sc Corner} connection pieces.}\label{fig:lin_corn}
\end{figure}

\subsection{Proof of Theorem \ref{thm:game_pspace}}
Now that we have finished describing the components of the reduction, we are ready to prove the main theorem (that {\sc 2048-Game} is PSPACE-Complete);

\begin{proof}[Proof of Theorem \ref{thm:game_pspace}]
Immediate from the construction described in Section \ref{sec:reduction}. Specifically, we first find a valid four-coloring of the NCL constraint graph $G$ using the algorithm
due to Robertson, et. al. \cite{rob97}, and then follow the procedure outlined in Section \ref{sec:planar} to connect the NCL vertex gadgets together using a combination of
{\sc $k$-$k'$-Connection, $k$-Line}, and {\sc $k$-Corner} pieces. 

Next, we must orient the edes of the constraint graph $G$ according to its current configuration. We set the tiles in the vertex gadgets accordingly; if an edge corresponds to an
{\it activated} configuration, the appropriate tile value is set on the vertex gadget (either at {\sf A}, {\sf B}, or {\sf C}).

We have now shown how to convert an orientation of an arbitrary planar constraint graph into an instance of {\sc 2048-Game}; if we can decide {\sc 2048-Game}, then
we can decide {\sc Config-to-Config}. Thus, {\sc 2048-Game} is PSPACE-Hard, and by Lemma \ref{lem:game_in_pspace}, $\text{\sc 2048-Game} \in \text{PSPACE}$. Therefore,
we conclude that {\sc 2048-Game} is PSPACE-Complete.

\end{proof}

\section{An FPT Algorithm for 2048-$k$-M{\fakesmallcaps OVES}}\label{sec:game_fpt}
On one hand, our paper shows that the most natural problem that emerges from 2048 is intractable. However, we show in this section
 that another variant of the game is fixed-parameter tractable (FPT). Recall that {\sc 2048-$k$-Moves} is the problem of deciding whether or not there exists a sequence of moves
 of length $k$ $\mathcal{S} \in \moveset^k$ such that after each of the moves is executed, the board will be in a winning configuration.
 
 This problem is decidable in polynomial time due to the fact that the number of moves is constant. While the constant may be quite large, it does not change the asymptotic running time
 of our algorithm. First, we introduce the notion of a {\it game tree};
 
 \begin{figure}[h]
 \begin{center}
 	\includegraphics[width=14cm]{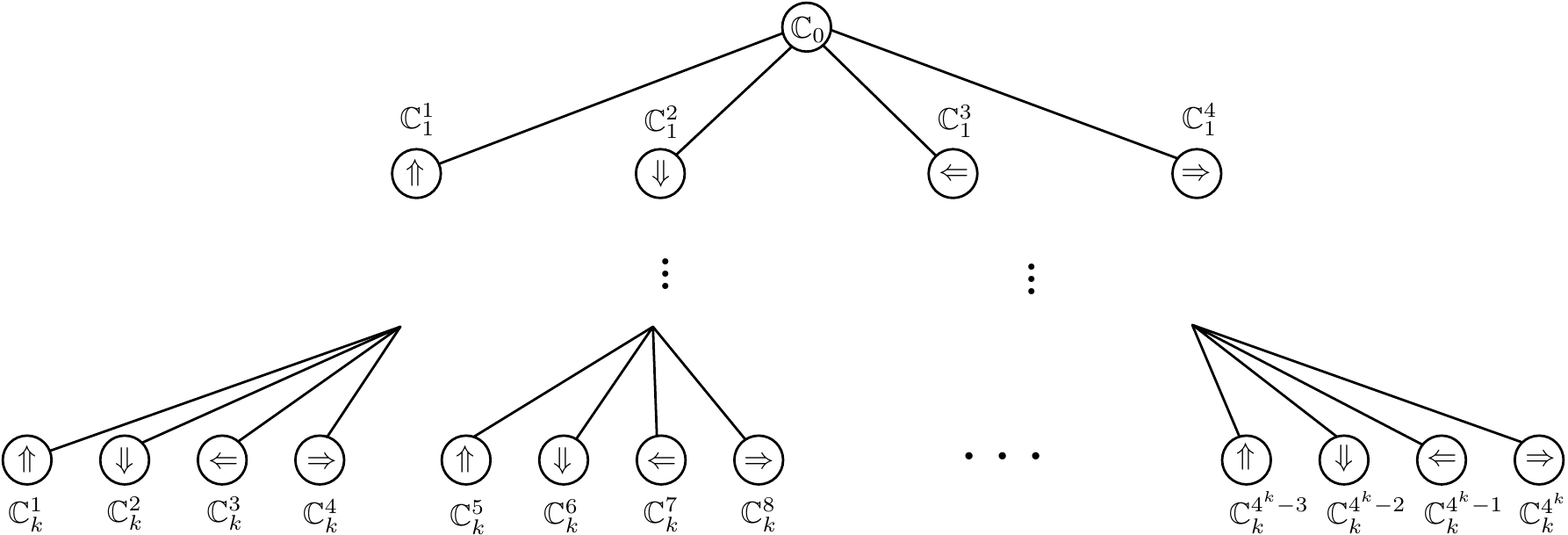}
 \end{center}
 \caption{A game tree for an instance of {\sc 2048-$k$-Moves} of depth $k$.}
 \end{figure}
 
 Each vertex of the tree that can be reached by a path of length $\ell$ represents a possible configuration of the game after $\ell$ moves. The tree can be constructed in $O(4^k)$ time
 by querying the oracle for the $4^k$ possible sequences of moves $\mathcal{S} \in \moveset^k$.
 Our algorithm, in essence, traverses the vertices of the game tree and halts when a winning configuration is found, or if there are none. We are now ready to prove the second
 result of our paper;
 
 \begin{thm}\label{thm:game_fpt}
 Given an $n \times n$ game board $\G$, {\sc 2048-$k$-Moves} is solvable in $O(4^k n^2)$ time.
 \end{thm}
 \begin{proof}
 The algorithm can perform any recursive traversal of the game tree to explore the leaf nodes. We first observe that the game tree is in fact a $4$-ary tree. It is well-known that
 for a $d$-ary tree of depth $\ell$, the number of leaf nodes is $d^\ell$. Thus, for our $4$-ary tree of depth $k$, there are $4^k$ leaf nodes.
 
 When each leaf is visited, the value of each of the $n^2$ tiles on the game board is compared
 to $2^m$, the goal tile. If they are equal, halt and output {\sf YES} and the configuration $\C_f$. If a game board is in a game over configuration (i.e. the subsequent configurations of the four
 children nodes are precisely equal to the current configuration), terminate the recursion for that branch of the tree. 
 
 Since we check each of the $n^2$ tiles of $\G$ at each of the $4^k$ leaf nodes, the worst-case running time of the procedure is $O(4^k n^2)$, showing that $\text{\sc 2048-}k\text{\sc -Moves} \in \text{FPT}$.
 \end{proof}

\bibliographystyle{amsplain}
\bibliography{2048-PSPACE}

\appendix

\section{Activation Sequences for A{\fakesmallcaps ND} and O{\fakesmallcaps R} Gadgets}\label{app:seq}

\begin{figure}[h]

\centering
\begin{tabular}{c c c c}
\includegraphics[width=2.75cm]{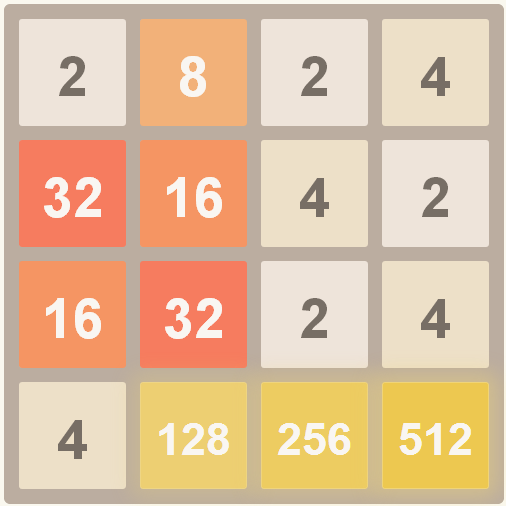} & \includegraphics[width=2.75cm]{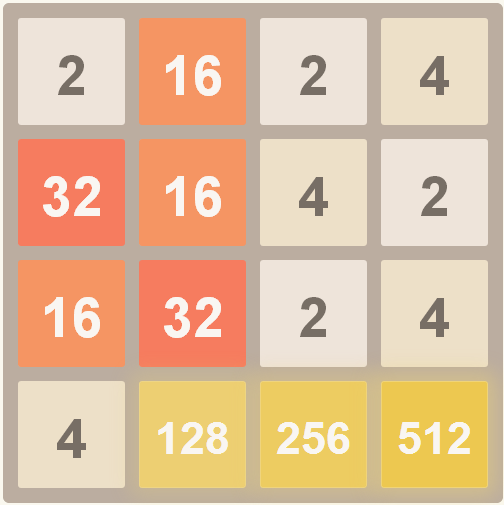} & \includegraphics[width=2.75cm]{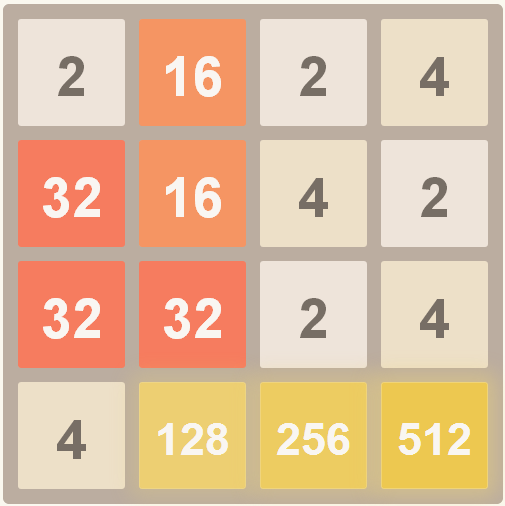} &  \includegraphics[width=2.75cm]{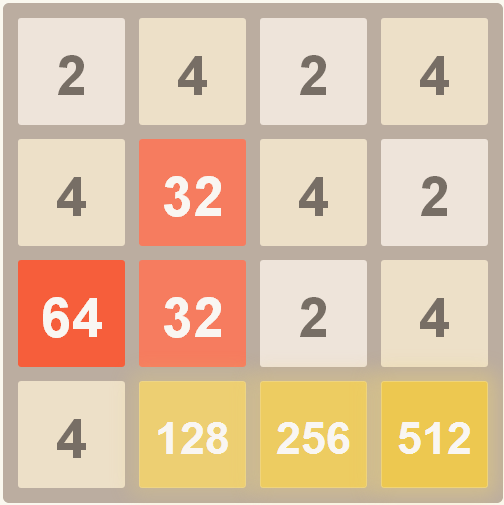}  \\
{\sf Initial State} &  $\Downarrow$ & $\Rightarrow$ & $\Downarrow$ \\
\includegraphics[width=2.75cm]{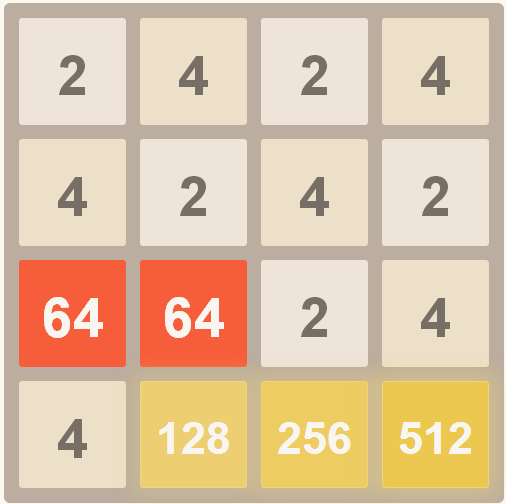} & \includegraphics[width=2.75cm]{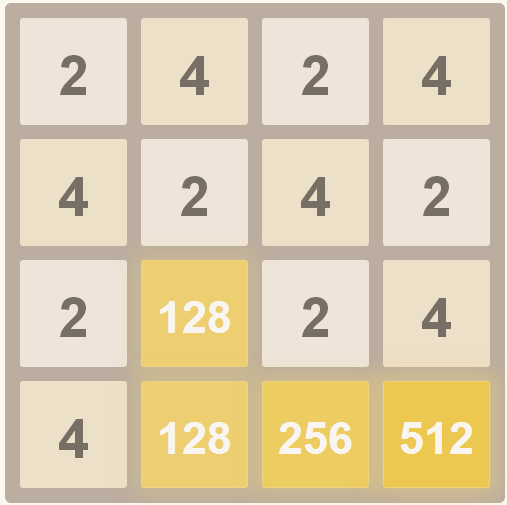} & \includegraphics[width=2.75cm]{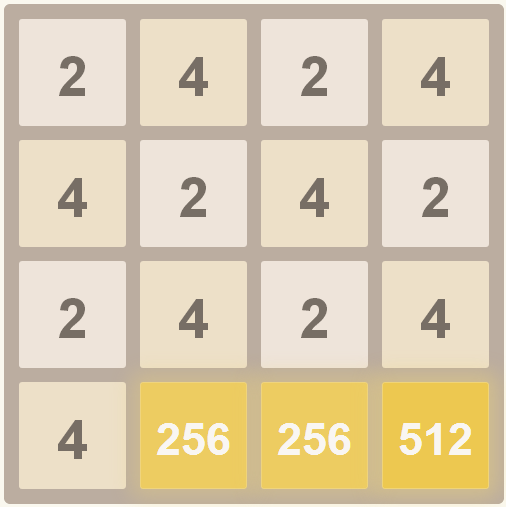} &  \includegraphics[width=2.75cm]{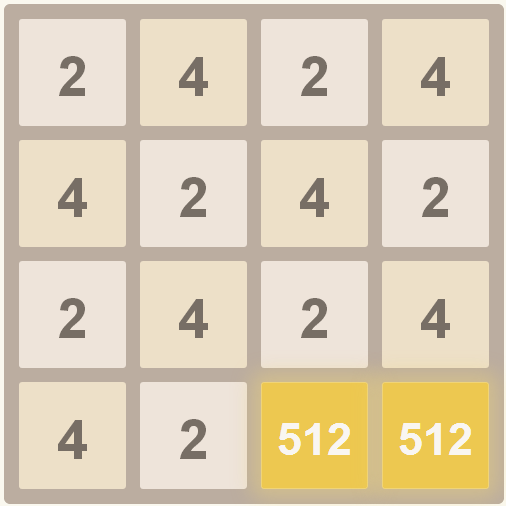}  \\
$\Downarrow$ & $\Rightarrow$ & $\Downarrow$ & $\Rightarrow$ \\
\multicolumn{4}{c}{\includegraphics[width=2.75cm]{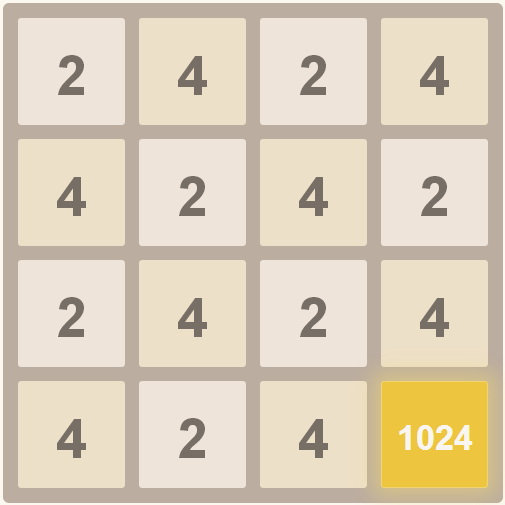}} \\
\multicolumn{4}{c}{$\Rightarrow$}

\end{tabular}
\caption{A possible sequence for activating {\sf C} in a $4$-{\sc And} gadget ($\Downarrow, \Rightarrow, \Downarrow, \Downarrow, \Rightarrow, \Downarrow, \Rightarrow, \Rightarrow$).}\label{fig:and_path}
\end{figure}

\begin{figure}[h]
\centering
\begin{tabular}{c c c c}
\includegraphics[width=2.75cm]{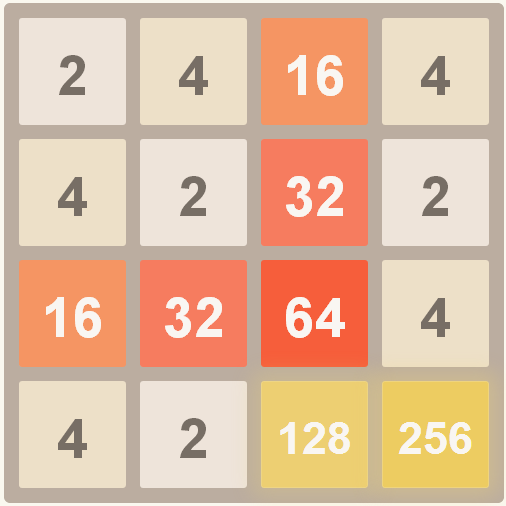} & \includegraphics[width=2.75cm]{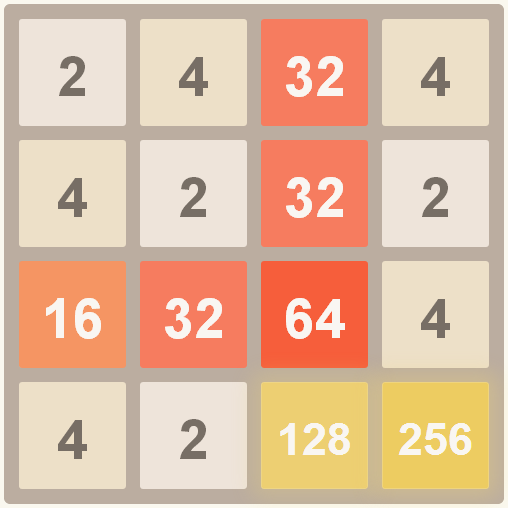} & \includegraphics[width=2.75cm]{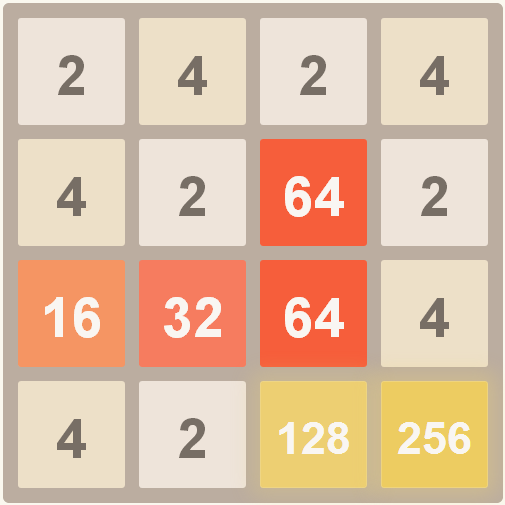} &  \includegraphics[width=2.75cm]{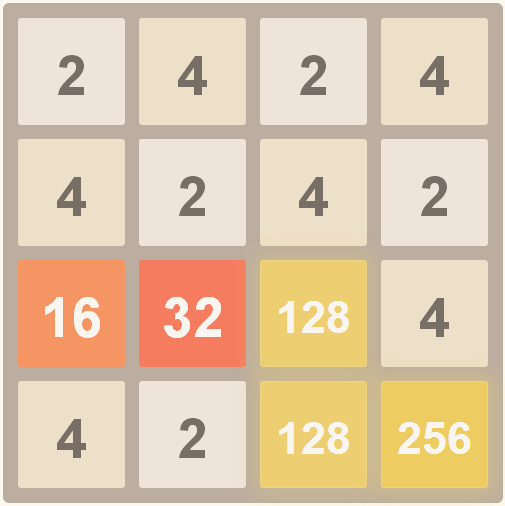}  \\
{\sf Initial State} &  $\Downarrow$(b) & $\Downarrow$(b) & $\Downarrow$(b) \\
\multicolumn{2}{c}{\includegraphics[width=2.75cm]{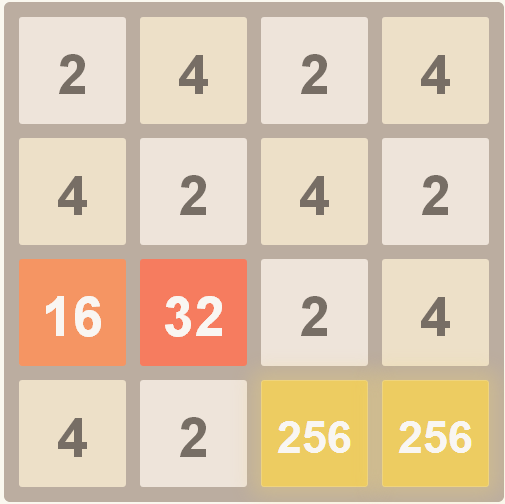}} & \multicolumn{2}{c}{\includegraphics[width=2.75cm]{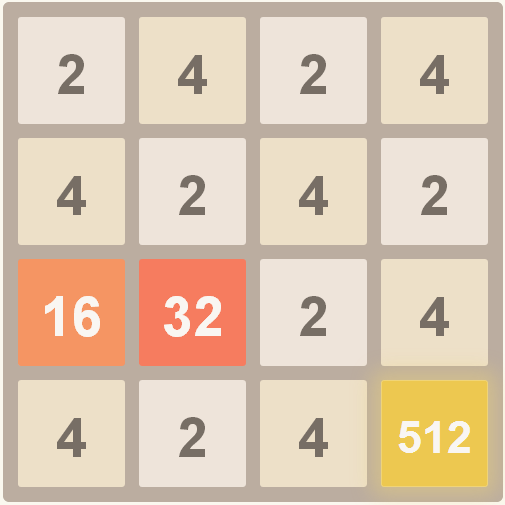}} \\
\multicolumn{2}{c}{$\Downarrow$(b)} &\multicolumn{2}{c}{$\Rightarrow$(b)} \\

\includegraphics[width=2.75cm]{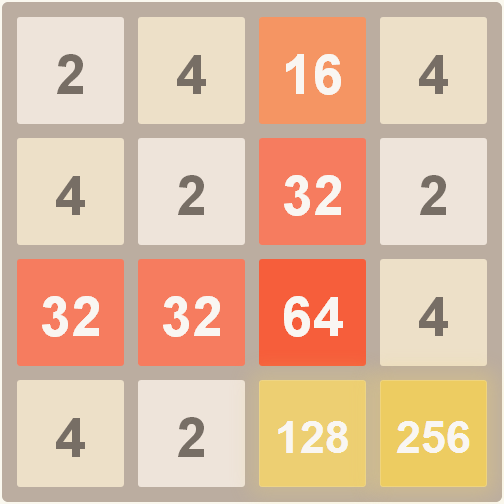} &  \includegraphics[width=2.75cm]{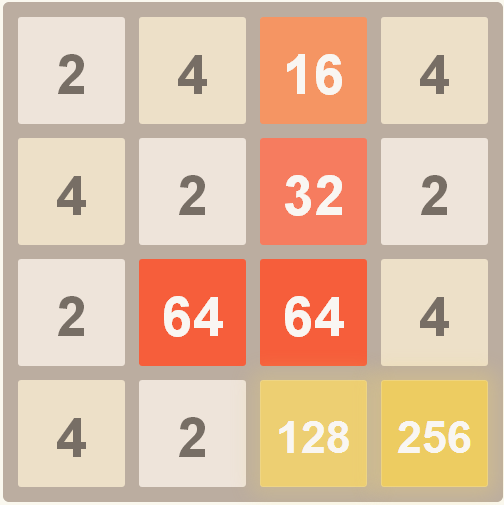} &  \includegraphics[width=2.75cm]{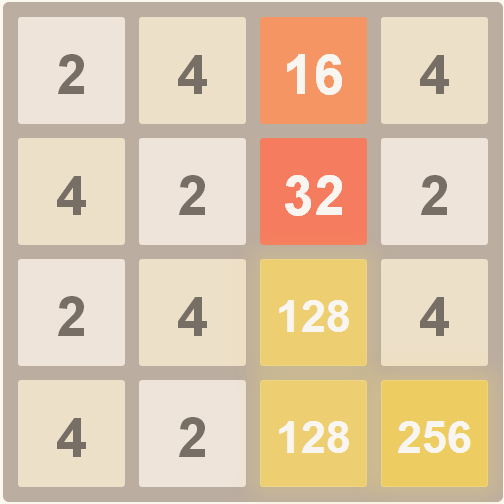} &  \includegraphics[width=2.75cm]{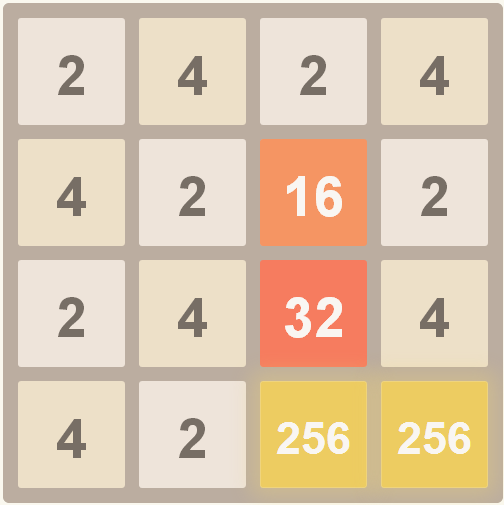}\\
$\Rightarrow$(a) & $\Rightarrow$(a) & $\Rightarrow$(a) & $\Downarrow$(a) \\
\multicolumn{4}{c}{\includegraphics[width=2.75cm]{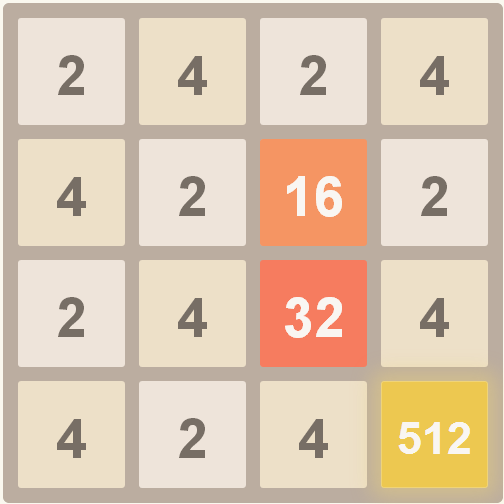}} \\
\multicolumn{4}{c}{$\Rightarrow$(a)}

\end{tabular}
\caption{Two possible sequences of moves to activate {\sf C}; (b) $\Downarrow, \Downarrow, \Downarrow, \Downarrow, \Rightarrow$, and 
(a) $\Rightarrow, \Rightarrow, \Rightarrow, \Downarrow, \Rightarrow$ (corresponding to {\sf A} and {\sf B} as the activating tile, respectively).}\label{fig:or_path}
\end{figure}

\end{document}